\documentclass{article}
\usepackage{algorithm,algpseudocode}
\usepackage[pdfpagelayout=OneColumn]{hyperref}
\usepackage{fancybox,pst-all}

\makeatletter

\def\@abssec#1{\vspace{.05in} \parindent 0in {\bf #1. }\ignorespaces}
\def\abstract{\@abssec{Abstract}}

\def\keywords{\@abssec{Key words}}

\newtheorem{lemma}{Lemma}

\newtheorem{proposition}{Proposition}

\def\combi#1#2{\left(\begin{array}{c} {#1} \\ {#2}\end{array}\right)}
\newcommand{\argmax}{\mathop{\mathrm{argmax}}}

\newcommand{\StatexIndent}[1][3]{%
  \setlength\@tempdima{\algorithmicindent}%
  \Statex\hskip\dimexpr#1\@tempdima\relax}

\algdef{S}[IF]{IfNoThen}[1]{\algorithmicif\ #1}%
\algnewcommand{\IIf}[1]{\State\algorithmicif\ #1\ \algorithmicthen}
\algnewcommand{\EndIIf}{\unskip\ \algorithmicend\ \algorithmicif}

\begin{document}
\title{Uniform random generations and rejection method(I) with binomial majorant}
\author{Laurent Alonso\thanks{INRIA Nancy-Grand Est,
    615, rue du Jardin Botanique,
    54600 Vandoeuvre-l\`es-Nancy, France.
    Email: {\tt Laurent.Alonso@inria.fr}
}}
\maketitle

\begin{abstract}
  We present three simple algorithms to uniformly generate `Fibonacci words' (i.e., some words that are enumerated by Fibonacci numbers), Schröder trees of size $n$ and Motzkin left factors of size $n$ and final height $h$. These algorithms have an average complexity of $O(n)$ in the unit-cost RAM model\footnote{Since we only present algorithms that use small integers in this article, we will use the unit cost RAM model to give the complexity, you can get the logarithmic cost RAM model by multiplying this complexity by a factor $log\,n$.} and use only small integers (integers smaller than $n^2$).
\end{abstract}


\section{Introduction}
One of the first generic methods used to uniformly generate various data structures is the {\em recursive method}\cite{Nijen1978v}\cite{Flajolet1994ACF}. It consists in computing the number of structures of different sizes and using these numbers to recursively generate an element. A drawback is that the coefficients are often very large: $\approx C^n$ for some constants $C$ (and so required to link to a library that can compute very large numbers) and the method is very slow. Denise et al.\cite{DENISE1999233} shows that it can be improved by performing lazy floating point operations, this allows to obtain a more efficient algorithm\footnote{This article presents algorithms corresponding to some context-free grammars, the average complexity $O(n^{1+\epsilon})$.} but it remains very difficult to implement.

For some problems: `Fibonacci words', Schröder trees\cite{PENAUD2001345}, Motzkin words\cite{ALONSO1994529},\cite{BACHER201742}, \cite{GouyouBeauchamps2010RandomGU} $\ldots$ rejection methods give efficient algorithms; in general, these methods are based on ingenious one-to-one correspondences and are easy to implement\footnote{In fact, generating a random Motzkin word gives a simple algorithm but concerning other problems, with low probability, these methods require computing floats with very high accuracy, $\ldots$}.
  
We try in this paper to extend the method proposed by \cite{ALONSO1994529}: this method first chooses the number of horizontal steps $m$ with an adequate probability using a rejection method, then generates a random Motzkin word with $m$ of horizontal steps (which is a simple problem). In~\cite{GouyouBeauchamps2010RandomGU}, Gouyou-Beauchamps et all show that this method can be extended to generate colored unary-binary trees with linear average complexity and partial injection with average complexity in $O(n^{5/4})$. Their method consists in choosing a number $m$ with probability $\frac{F_n(m)}{\sum_i F_n(i)}$ by finding a distribution ${\overline B}_n(m)$ such that :
\begin{itemize}
\item it is easy to generate a number $m$ with probability $\frac{{\overline B}_n(m)}{\sum_i {\overline B}_n(i)}$,
\item we have for all $m$, $F_n(m)\leq {\overline B}_n(m)$ and we have an efficient algorithm to accept a value of $m$ with probability $\frac{F_n(m)}{{\overline B}_n(m)}$.
\end{itemize}
This method is quite general, the proposed algorithms are easy to implement and ingenious; however finding a distribution ${\overline B}_n(m)$ that fulfills these properties is not straightforward even if they propose to use for ${\overline B}_n(m)$ a binomial distribution multiplied by a constant.

In this paper, we show that the choice of a binomial distribution (up to a minor modification) close to the final distribution allows to obtain simply an efficient algorithm at least for some basic structures. More precisely, we will show that it is possible to obtain efficient algorithms to randomly and uniformly generate a `Fibonacci word' of size $n$ (the simplest problem), a Schröder tree of size $n$ (with almost the same algorithm) and finally a Motzkin left factor of size $n$ and final height $h$. As in~\cite{GouyouBeauchamps2010RandomGU}, these algorithms use only small integers (i.e. integers smaller than $n^2$), have an average complexity of $O(n)$ and can be implemented very easily. Another interesting property is that almost all these algorithms use on average only $O(\sqrt n)$ operations that involve numbers greater than $n$; only the generation of a Motzkin left factor of size $n$ and final height $h$ when $h\geq n-\frac 1 2-\frac{\sqrt{8n+9}}2$ is different because it uses on average $O(n)$ operations that all involve numbers less than $3n+7$.

In a future paper\cite{AlonsoRandomII}, we will propose another method which consists in drawing the initial value of $m$ differently: it is enough to generate an integer $m$ with the basic uniform distribution (i.e. with probability in $\frac 1 n$). We will show that, at least for simple problems, this leads to very efficient algorithms. Although we will probably choose to use this second method for simple problems, we think that these two papers are interesting because :

\begin{itemize}
\item together, they suggest that the initial values of $m$ can be chosen by different types of methods,
\item these solutions have different properties: in this paper, we propose to guess the first value of $m$ with a distribution close to the final distribution, which leads to few rejections but to many operations in the main loop. The new paper will propose the opposite, so we will have a very high rejection rate while the time spent in a loop is very low. We suspect that in some cases, for example if we want to find an algorithm to choose the number of binary nodes and the number of ternary nodes in $3-ary$ trees, having a high rejection rate is prohibitive,
\item in practice, it is possible to mix the ideas of these two articles: \cite{AlonsoRandomII} to generate a first value of $m$ with a distribution close to the final distribution but still following the methods proposed in this article.
\end{itemize}

In Section~\ref{abstract}, we introduced the notations, an abstract algorithm and the different generating functions that will be used to generate the initial $m$ value. Then, in the Sections~\ref{fibonacci}--\ref{motzkin}, we will see how to instantiate this abstract algorithm to generate `Fibonacci words', Schröder trees of size $n$ and a Motzkin left factor of size $n$ and given final height $h$ and prove that these algorithms choose a value $m$ with an adequate probability and have an average complexity of $O(n)$. 

Finally, after a little discussion, we give in the appendices the proofs of many lemmas. These proofs use only elementary mathematics (and had to be done); as writing detailed proofs would probably be too long and boring, we tried to give just enough details so that they can be checked without too much trouble...

\section{Notations and Abstract Algorithm}\label{abstract}
\subsection{Notations}
In the following we will use:
\begin{itemize}
\item $[P]$ the Kenneth Iverson's convention\cite{Graham1988ConcreteM}(page 24): $[P]=1$ if $P$ is true and $[P]=0$ if $P$ is false,
\item $random(k)$ the function that returns an integer between $0$ and $k-1$ with probability $\frac 1 k$\footnote{to simplify some notations, when $p$,$q$,$r$,$s$ are four integers, we will sometimes note $random(p/q)\geq r/s$ instead of $random(p\,s)\geq q\,r$, $\ldots$},
\end{itemize}
and:
\begin{itemize}
\item $F(m)$ will count the final number of elements of size $m$,
\item $M=\min \argmax_m F(m)$: the first value of $m$ that maximizes $F(m)$,
\item $\underline{B}(m)$ is the number of instances of the initial generator that correspond to a choice $m$,
\item $B$ a slightly modified function of $\underline{B}$ to impose for example $B(M-1)=B(M)$, $\ldots$,
\item $\overline{B}(m)=B(m) \frac {F(M)}{B(M)}$ : the normalized function such that $\overline{B}(M)=F(M)$,
\item ${\underline G}(m)=\frac{F(m)}{\underline{B}(m)}$, $G(m)=\frac{F(m)}{B(m)}$, ${\overline G}(m)=\frac{F(m)}{\overline{B}(m)}$.
\end{itemize}
Finally, we will note:
\begin{itemize}
\item ${\cal R}F(m)=F(m+1)/F(m)$ the ratio between two consecutive terms of $F$ when they are defined,
\item if ${\cal R}F(m)={\text numerator}(m)/{\text denominator}(m)$ is a ratio between two polynomials in $m$, $f(x)={\text numerator}(x)-{\text denominator}(x)$,
\end{itemize}
and we will define ${\cal R}{\underline B}(m)$, $\ldots$, ${\overline g}(x)$ in the same way.

\subsection{An Abstract Algorithm}
With these notations, we will use the following procedure to draw a number $m$ with probability $\frac {F(m)}{\sum_m F(m)}$:
\begin{algorithm}
\caption{Choose $m$ with probability $\frac {F(m)}{\sum_m F(m)}$:}\label{algAbstract}
\begin{algorithmic}[1]
  \Function{accept\_m}{$M$,$m$} \Comment{accept $m$ with probability ${\overline G}(m)=\frac{F(m)}{{\overline B}(m)}$}
  \If {$m<M$}
  \ForAll{$i \in [m,M-1]$} \Comment{accept with probability $\frac{{\cal R}{\overline B}(i)}{{\cal R}F(i)}$}
  \IIf{$random({\cal R}F(i))\geq {\cal R}{\overline B}(i)$} return FALSE \EndIIf
  \EndFor
  \Else
  \ForAll{$i \in [M,m-1]$} \Comment{accept with probability $\frac{{\cal R}F(i)}{{\cal R}{\overline B}(i)}$}
  \IIf{$random({\cal R}{\overline B}(i))\geq {\cal R}F(i)$} return FALSE \EndIIf
  \EndFor
  \EndIf
  \State return TRUE
  \EndFunction
  \item[]
  \State $M \gets \text{first value of $m$ which maximalizes $F(m)$}$
  \While{true}
  \State choose $m$ with probability $\frac{B(m)}{\sum_m{B(m)}}=\frac{\overline{B}(m)}{\sum_m \overline{B}(m)}$
  \IIf{$\textsc{accept\_m}(M,m)$} return $m$\EndIIf
  \EndWhile
\end{algorithmic}
\end{algorithm}

Lines 1 to 12 correspond to the main function that accepts a value $m$ with probability $\frac{F(m)}{\overline{B}(m)}$. It is based on the following relations, when $m<M$:
\begin{eqnarray*}
  \frac{F(m)}{\overline{B}(m)}&=&\frac{F(m)\overline{B}(m+1)}{\overline{B}(m)F(m+1)} \frac{F(m+1)\overline{B}(m+2)}{\overline{B}(m+1)F(m+2)}\ldots \frac{F(M-1)\overline{B}(M)}{\overline{B}(M-1)F(M)} \frac{F(M)}{\overline{B}(M)} \\
  &=& \frac{{\cal R}{\overline B}(m)}{{\cal R}F(m)}\frac{{\cal R}{\overline B}(m+1)}{{\cal R}F(m+1)}\ldots\frac{{\cal R}{\overline B}(M-1)}{{\cal R}F(M-1)}
\end{eqnarray*}
similarly when $m>M$:
\[\frac{F(m)}{\overline{B}(m)}=\frac{{\cal R}F(M)}{{\cal R}{\overline B}(M)}\frac{{\cal R}F(M+1)}{{\cal R}{\overline B}(M+1)}\ldots \frac{{\cal R}F(m-1)}{{\cal R}{\overline B}(m-1)},\]
and finally to $\frac{F(M)}{\overline{B}(M)}=1$.
It therefore assumes that when $i<M$, ${\cal R}F(i)\geq {\cal R}{\overline B}(i)$ and when $i\geq M$, ${\cal R}F(i)\leq {\cal R}{\overline B}(i)$.

Lines 13 to 17 group the main lines of the program: First, we compute the value of $M$ (line 13), then we choose a value of $m$ by repeatedly drawing a value of $m$ with probability $\frac{B(m)}{\sum_m{B(m)}}=\frac{\overline{B}(m)}{\sum_m \overline{B}(m)}$ and accept it with probability $\frac{F(m)}{{\overline B}(m)}$. We accept in this loop a value $m$ with the probability $\frac{\sum_m F(m)}{\sum_m \overline{B}(m)}$, this value must be big enough to obtain an efficient algorithm. In addition, we have more specifically:
\begin{lemma}\label{complexity}
  Let us note $C_1$ the complexity necessary to choose a value $m$ with probability $\frac{\overline{B}(m)}{\sum_m \overline{B}(m)}=\frac{B(m)}{\sum_m B(m)}$, $C_2=\frac{\sum_m |M-m| \overline{B}(m)}{\sum_m \overline{B}(m)}=\frac{\sum_m |M-m| B(m)}{\sum_m B(m)}$ the mean absolute difference of $B$. If we assume that the tests of lines~4 and 8 can be performed with complexity $O(1)$, the average complexity of this algorithm is $O\left((1+C_1+C_2)\,\frac{F(M)}{\sum_m F(m)} \frac{\sum_m B(m)}{B(M)}\right)$.

  Moreover, if the algorithm leads to a value $m$, the average complexity \footnote{This notion of complexity is important because for some algorithms, its maximum value can be very different from the average complexity $\ldots$.} is $O\left(|M-m|+(1+C_1+C_2)\,\frac{F(M)}{\sum_m F(m)} \frac{\sum_m B(m)}{B(M)}\right)$.

  In the particular case where $C_1=O(n)$, $C_2=O(\sqrt n)$, $\frac{\sum_m F(m)}{F(M)}=\Omega(\sqrt n)$, $\frac{\sum_m B(m)}{B(M)}=O(\sqrt n)$, we will obtain an average complexity in $O(n)$ and when the algorithm ends with a value $m$ also $O(n)$.
\end{lemma}
\begin{proof}{Proof}
  We first study the number of operations passed inside a loop. We need by definition, at line~15, $C_1$ operations and as {\sc accept\_m} is called with a number $m$ (chosen with probability $\frac{B(m)}{\sum_m B(m)}$) at line~16: $O(1)+C_2$ operations. This corresponds to an average complexity of a main loop execution of $O(1)+C_1+C_2$.

  We now know that it accepts a value $m$ with the probability $C_3=\frac{\sum_m F(m)}{\sum_m \overline{B}(m)}$. The average number of loops is therefore
  \[1+(1-C_3)+(1-C_3)^2+\ldots = \frac 1 {1-(1-C_3)}=\frac{\sum_m \overline{B}(m)}{\sum_m F(m)}=\frac{F(M)}{\sum_m F(m)} \frac{\sum_m B(m)}{B(M)}.\]

  We finish by noting that when the algorithm returns a value $m$, we use in the last execution of the loop $O(C_1+1+|m-M|)$ operations.
\end{proof}

\subsection{Some generating functions}

We present three generating functions that we will use in this paper to generate an initial value $m$.
\paragraph{`Basic Generator'}
This is the simplest generator: to generate an integer between $[0,k]$ with probability $\frac 1 {k+1}$, we use $random(k+1)$. So we have $\underline{B}^k(m)=B^k(m)=1$ when $0\leq m\leq k$ (and ${\cal R}\underline{B}^k(m)={\cal R}B^k(m)=1$ when $0\leq m<k$) Its complexity is very low: $O(1)$. We will use it to generate a Motzkin left factor of final height $h$ when $h$ is near $n$.
\paragraph{`Binomial Generator'}
We choose ${\underline B}_M(m)=\combi{2M}m$, $B_M(m)={\underline B}_M(m)$ when $m\neq M$ and $B_M(M)=B_M(M+1)$. These choices give us an efficient algorithm to draw the first value of $m$: draw $2M$ random value $0$ or $1$, let $m$ count the number of $0$, finally, if we find $m=M$, accept this choice with probability $\frac M {M+1}$. We also obtain:
\begin{eqnarray*}
  {\cal R}{\underline B}_M(m) &=& \frac{2M-m}{m+1}\\
  {\cal R}B_M(m) &=& {\cal R}{\underline B}_M(m) \text{ if }m<M-1\text{ or }m>M \\
  {\cal R}B_M(M-1) &=& {\cal R}B_M(M) = 1
\end{eqnarray*}

This generator chooses an integer $m$ between $0$ and $2M$ with a probability close to the binomial distribution. It has an average complexity of $O(M)$. We will need the following lemma to relate $B_M(M)$ to $\sum_m B_M(m)$.
\begin{lemma}\label{lemmaG1}
We have when $M\geq 1$: $\frac{\sum |M-m|B_M(m)}{\sum B_M(m)}=\Theta(\sqrt M)$ and $\frac{\sum B_M(m)}{B_M(M)}= \Theta(\sqrt M)$.
\end{lemma}
\begin{proof}{Proof}
  ${\underline B}_M$ is the classical binomial distribution with $p=q=\frac 1 2$ multiplied by $2^{2M}$, so we have: $\frac{\sum |M-m|{\underline B}_M(m)}{\sum {\underline B}_M(m)}=\Theta(\sqrt M)$ and using the Stirling Formula $\frac{\sum {\underline B_M}(m)}{{\underline B}_M(M)}= \Theta(\sqrt M)$.

  We can conclude by noting that when $m\neq M$, ${\underline B}_M(m)=B_M(m)$ and $B_M(M)=\frac m {m+1} {\underline B}_M(M)$ (and thus $B_M(M)=\Theta({\underline B}_M(M))$ and $\sum B_M(m)=\Theta(\sum {\underline B_M}(m))$).
\end{proof}

This algorithm can be coded as:
\begin{algorithmic}
  \Function{bin}{$M$} \Comment{choose $m$ with probability $\frac{B_M(m)}{\sum_m{B_M(m)}}$}
  \While{true}
  \State $m \gets 0$
  \ForAll{$i \in [1,2M]$}
  \IIf{$random(k+1)=0$} $m \gets m+1$ 
  \EndFor
  \IIf{$m=M$ AND $random(M+1)=0$} continue \EndIIf
  \State return $m$
  \EndWhile
  \EndFunction
\end{algorithmic}
We will use the `Binomial Generator' when possible but note that it has the limitation of only being able to draw a number between $0$ and $2M$(which sometimes is not enough).

\paragraph{`Extended Binomial Generator'}
The previous generator corresponds to an equiprobable binomial distribution: $p=q=\frac 1 2$, we will use here an equivalent of the generic distribution with $p=\frac 1 {k+1}$ and $q=1-p$; let's define ${\underline B}_M(m)$ as:
\[{\underline B}_M^{k,\alpha}(m)=k^{(k+1)M+\alpha-m}\combi{(k+1)M+\alpha}m,\]
we get ${\cal R}{\underline B}_M^{k,\alpha}(m)=\frac{(k+1)\,M+\alpha-m}{k(m+1)}$. If $0\leq \alpha<k$, this function reaches its maximum when $m=M$, it can draw numbers between $0$ and $(k+1)M+\alpha$ and it is simple to implement: just draw $(k+1)M+\alpha$ random integers between $0$ and $k$, and count the number of $0$.

We can now define $B(m)$ as:
\begin{itemize}
\item $B_M^{k,\alpha}(m)={\underline B}_M^{k,\alpha}(m)$ when $m\neq M$,
\item $B_M^{k,\alpha}(M)=max({\underline B}_M^{k,\alpha}(M-1),{\underline B}_M^{k,\alpha}(M+1))$.
\end{itemize}
This gives us:
\begin{itemize}
  \item when $m<M-1$ or $m\geq M+1$, ${\cal R}B^{k,\alpha}_M(m) = {\cal R}{\underline B}^{k,\alpha}_M(m)$,
  \item when $B^{k,\alpha}_M(M)=B^{k,\alpha}_M(M+1)$
    \[{\cal R}B^{k,\alpha}_M(M-1) = \frac{(k M+\alpha+1)(k M+\alpha)}{k^2M(M+1)}, {\cal R}B^{k,\alpha}_M(M)=1,\]
  \item and when $B^{k,\alpha}_M(M)=B^{k,\alpha}_M(M-1)$
    \[{\cal R}B^{k,\alpha}_M(M-1) = 1, {\cal R}B^{k,\alpha}_M(M)=\frac{(k M+\alpha+1)(k M+\alpha)}{k^2M(M+1)}.\]
\end{itemize}
We will prove the following lemma in the Appendix~\ref{proofG2}:
\begin{lemma}\label{lemmaG2}
  When $0\leq \alpha \leq k-1$, $k\geq 1$ and $M\geq 1$, we have:
  \begin{itemize}
  \item ${\underline B}_M^{k,\alpha}(m)\geq {\underline B}_M^{k,\alpha}(m+1)$ iff $m\geq M$,
  \item $\frac{\sum_m B_M^{k,\alpha}(m)}{B_M^{k,\alpha}(M)}=\Theta(\sqrt M)$,
  \item $\frac{\sum_m |m-M|B_M^{k,\alpha}(m)}{\sum_m B_M^{k,\alpha}(m)}=O(\sqrt M)$.
  \end{itemize}
\end{lemma}

When $0\leq \alpha\leq k-1$, this algorithm can be coded as:
\begin{algorithmic}
  \Function{extended\_bin}{$M,k,\alpha,which$}
  \State {\bf Require:} $M > 0$, $k$, $0 \leq \alpha \leq k-1$
  \While{true}
  \State $m \gets 0$
  \ForAll{$i \in [1,(k+1)M+\alpha]$}
  \IIf{$random(k+1)=0$} $m \gets m+1$ \EndIIf
  \EndFor
  \IIf{$m\neq M$} return $m$ \EndIIf
  \If{$which=\text{'M+1'}\textsc{ and }random(k(M+1))\geq k\,M+\alpha$}
  \State continue \Comment{reject $m$}
  \ElsIf{$which=\text{'M-1'}\textsc{ and }random(k\,M+\alpha+1)\geq k\,M$}
  \State continue \Comment{reject $m$}
  \EndIf

  \State return $m$
  \EndWhile
  \EndFunction
\end{algorithmic}
It generates an integer $m$ with probability $\frac{B_M^{k,\alpha}(m)}{\sum_{i=0}^{(k+1)M+\alpha}B_M^{k,\alpha}(i)}$ and its average complexity is $O(k(M+1)+\alpha)$.

\section{Generation of a `Fibonacci word'}\label{fibonacci}
Fibonacci numbers have a very long history: they correspond to the sequence A000045 in The On-Line Encyclopedia of Integer Sequences\cite{oeis} and are defined by the equation $F_n=F_{n-1}+F_{n-2}$ with $F_0=F_1=1$, they are related to the golden number $\varphi=\frac{1+\sqrt 5} 2$ by the well known relation $F_n=\frac{\varphi^{n+1}-(1-\varphi)^{-(n+1)}} {\sqrt 5}=\left\lceil \frac {\varphi^{n+1}} {\sqrt 5}\right\rfloor$ (using the nearest rounding function).

Let us note first that we can obtain a simple algorithm by recursively choosing a letter $a$ with probability $\frac{F_{n-1}}{F_n}\approx \varphi$ and a letter $b$ with probability $\frac{F_{n-2}}{F_n}\approx \varphi^2$, but this algorithm is not very efficient. So we prefer to use a rejection algorithm: we simply build a sequence of letters by choosing $a$ with probability $\alpha$ and $b$ with probability $\alpha^2$ with $\alpha+\alpha^2=1$ ($\alpha=\frac{-1+\sqrt 5} 2$) and we stop when we have obtained a word of ${\cal F}_n$ or a word of ${\cal F}_{n+1}$ (a failure). In this last case, we start a new generation from the beginning. We can check that in one try, we generate each word of ${\cal F}_n$ with a probability $\alpha^n$ and each word of ${\cal F}_{n-1}$ followed by a letter $b$ with a probability $\alpha^{n+1}$ ; the probability of failure is therefore $\frac{\alpha^{n-1}F_{n-1}\alpha^2}{\alpha^n F_n+\alpha^{n+1}F_{n-1}} \approx \frac \alpha {\varphi+\alpha} = 1-\frac {\sqrt 5+1}{2\sqrt 5}$. This proves that the average number of trials is $O\left(\frac{2\sqrt 5}{\sqrt 5+1}\right)=O(1)$ and that the average complexity is $O(n)$.

However, the main step of this algorithm is to generate a real $x$ in the interval $[0,1]$ and to compare it to $\alpha$. Thus, even if we start generating the most significant digits of $x$, we need to check if we have enough digits to be sure that $x<\alpha$ or $x\geq \alpha$. If this is not the case, we need to draw the following digits...

We present here an algorithm that draws a random word containing the letters $a$ and $b$ such that $n=|a|+2|b|$ but using only small integers (less than $n^2$). We have $F_n(m)=\combi{n-m}m$ where $m$ counts the number of letters $b$, which gives us:
${\cal R}F(m)=\frac{(n-2m)(n-1-2m)}{(m+1)(n-m)}$, $f(x)=(n-2x)(n-1-2x)-(x+1)(n-x)$, and the following lemma which groups together some necessary properties of $F_n(m)$ and is proved in the Appendix~\ref{proofFib}:
\begin{lemma}\label{lemma1}
  When $n>0$, 
  \begin{itemize}
  \item there exists a unique value $\tilde m$ in $[0 \ldots \lfloor n/2 \rfloor]$ such that $f(x)>0$ with $x\in [0,\lfloor n/2\rfloor]$ if and only if $x<\tilde m$,
  \item $\tilde m=\frac {5n-3-\sqrt{5n^2+10n+9}}{10}$, $M=\lceil \tilde m \rceil \approx \frac{5-\sqrt 5}{10}n$,
  \item $\frac{F_n}{F_n(M)}=\Omega(\sqrt n)$.
  \end{itemize}
\end{lemma}
The existence of $\tilde m$ gives us a linear method for computing $M$(Algorithm~\ref{algAbstract}, line 13): we can either make $m$ go from 0 to $\lfloor n/2\rfloor$ and stop as soon as ${\cal R}F_n(m)\leq 1$ (i.e. $f(m)\leq 0$) or we can compute $\lceil \tilde m \rceil$ directly.

We choose the `Binomial Generator' to obtain the first value of $m$(Algorithm~\ref{algAbstract}, line 15). We thus obtain:
\begin{eqnarray*}
  {\cal R}{\overline B}_M(m) &=& {\cal R} B_M(m) \\
  {\cal R}{\underline G}_M(m) &=& \frac{(n-2m)(n-1-2n)}{(2M-m)(n-m)} \\
  {\cal R}{\overline G}_M(m) &=& \frac{(n-2m)(n-1-2n)}{(2M-m+[m=M]-[m=M-1])(n-m)}.
\end{eqnarray*}

We now need the following lemma (which is proved in the Appendix~\ref{proofFib}):
\begin{lemma}\label{lemma2}
  We have:
  \begin{itemize}
  \item ${\cal R}{\overline B}_M(m)< {\cal R}F_n(m)$ if and only if $m < M$,
  \end{itemize}
\end{lemma}
to show that we accept a value of $m$ correctly with a probability of ${\cal R}{\overline G}_M(m)=\frac{{\cal R}F_n(m)}{{\cal R}{\overline B}_M(m)}$ and complete lines 4 and 8 of the Algorithm~\ref{algAbstract}.

Now as a consequence of Lemma~\ref{complexity}, $C_1=O(n)$, $C2=O(\sqrt n)$, \ldots, we obtain:
\begin{proposition}\label{proposition1} 
  This algorithm chooses $m$ with probability $\frac{F_n(m)}{F_n}$ and has an average complexity in $O(n)$ using only numbers with $2log_2(n)$ bits.
\end{proposition}

{\bf Remark:} We use numbers with $2log_2(n)$ bits:
\begin{itemize}
\item $n$ times to compute the value $M$ that maximizes $F_n^m$. But we can also compute a rough approximation $m'=\left\lceil \frac{5-\sqrt 5}{10} n\right\rceil$ and then search around $m'$ to find the value $M$ using $O(1)$ times such large numbers,
\item $|m-M|$ times to check the acceptance of the value $m$ with a probability of $F_n^m/{\overline B}_M(m)$. So, we can expect $|m-M|$ to be $O(\sqrt M)=O(\sqrt n)$.
\end{itemize}
This suggests that using a library optimized to compute with numbers of $2\log_2(n)$ bits or not will not change the average complexity of this algorithm much.

The final algorithm code can be written as in Algorithm~\ref{algFibo}.
\begin{algorithm}
\caption{Draw a random Fibonacci word with $|a|+2|b|=n$:}\label{algFibo}
\begin{algorithmic}
  \Function{accept\_m}{$n$,$M$,$m$} \Comment{accept $m$ with probability ${\overline G}(m)=\frac{F_n(m)}{{\overline B}_M(m)}$}
  \If {$m<M$}
  \ForAll{$i \in [m,M-1]$}
  \State $v \gets random((n-2i)(n-2i-1))$
  \IIf{$v\geq (n-i)(2M-i-[i=M-1])$} return FALSE\EndIIf
  \EndFor
  \Else
  \ForAll{$i \in [M,m-1]$}
  \State $v \gets random((n-i)(2M-i+[i=M]))$
  \IIf{$v\geq (n-2i)(n-2i-1)$} return FALSE\EndIIf
  \EndFor
  \EndIf
  \State return TRUE
  \EndFunction
  \item[]
  \Require $n > 0$
  \State $M \gets \text{first value of $m$ which maximalizes $F_n(m)$}$
  \While{true} \Comment{Choose $m$}
  \State $m \gets \textsc{bin}(M)$
  \IIf{$\textsc{accept\_m}(n,M,m)$} break\EndIIf
  \EndWhile
  \State \Comment{draw a random word with $n-2m$ letters $a$ and $m$ letters $b$.}
  \State $b \gets m$, $a \gets n-2m$
  \While{$a+b>0$}
  \If{$random(a+b)<a$}
  \State add a letter $a$, $a \gets a-1$,
  \Else
  \State add a letter $b$, $b \gets b-1$.
  \EndIf
  \EndWhile
\end{algorithmic}
\end{algorithm}

\section{Schröder numbers}\label{schorder}
The Schröder numbers are defined by the recurrence $S_0=1$, $S_1=2$, $S_n=3 S_{n-1}+\sum_{k=1}^{n-2} S_k S_{n-k-1}$. $S_n\approx \frac{\sqrt 2 +1}{2^{3/4}\sqrt{\pi}}\frac{(3+2\sqrt 2)^n}{n\sqrt n}$(see \cite{Au2019SomePA}) is also the number of paths $x(1,1)$, $\bar x(1,-1)$, $z(0,2)$ that go from $(0,0)$ to $(2n,0)$ and do not go below the x axis.

Penaud et al.\cite{PENAUD2001345} propose to use a rejection method. This method is more complicated than for the Fibonacci number because it is based on correspondences between the Schröder path of size $2n$ and a subset of prefixes of Schröder paths of size $2n+1$.  The prefixes of Schröder paths of size $2n+1$ are generated letter by letter: $x$ is chosen first, then $x$ is chosen with probability $\sqrt 2-1$, $\bar x$ with probability $\sqrt 2-1$, $z$ with probability $(\sqrt 2-1)^2=3-2\sqrt 2$, $\ldots$ and the path is rejected if it passes under the y-axis; then the matches are used to reject this path or convert it to a Schröder path. In summary, this method is simple but requires some time to compute the floats with high accuracy.

We have:
\[F_n = S_n = \sum_{m=0}^n \frac 1 {n+m+1} \combi{n+m+1}{m,m+1,n-m}=\sum_{m=0}^n F_n(m)\]
this formula corresponds for some $m$ to draw a random word with $m+1$ letters $x$, $m$ letters $\bar x$, $n-m$ letters $z$, then use the ``Cycle Lemma''\cite{DERSHOWITZ199035} which gives a $n+m+1$-1 correspondence between such words and the Schröder path of size $2n$ which contains $m$ letters $x$.

We can now use the `Binomial generator' as for the Fibonacci algorithm to draw a value $m$ with probability $\frac {F_n(m)} {F_n}$. When $m$ is chosen with an adequate probability, we finish by generating a random word with $m+1$ letters $x$, $m$ letters $\bar x$, $n-m$ letters $z$ and use the ``Cycle Lemma''\cite{DERSHOWITZ199035} to get the final Schröder path.

So we have:
\begin{eqnarray*}
  {\cal R}F_n(m) &=& \frac{(n+m+1)(n-m)}{(m+1)(m+2)} \\
  {\cal R}{\underline G}_M(m)&=&\frac{(n+m+1)(n-m)}{(m+2)(2M-m)} \\
  {\cal R}{\overline G}_M(m)&=&\frac{(n+m+1)(n-m)}{(m+2)(2M-m+[m=M]-[m=M-1])}
\end{eqnarray*}

We will finally need a lemma which is proved in the Appendix~\ref{proofS1}:
\begin{lemma}\label{lemmaS1}
  When $n>1$, 
  \begin{itemize}
  \item there exists a unique value $\tilde m$ in $[0 \ldots n]$ such that $f(x)>0$ with $x\in [0,n]$ if and only if $x<\tilde m$,
  \item $\tilde m=-1+\frac {\sqrt{2n^2+2n}}{2}$, $M=\lceil \tilde m \rceil \approx \frac n {\sqrt 2}$,
  \item $\frac{F_n}{F_n(M)}=\Omega(\sqrt n)$,
  \item ${\cal R}{\overline B}_M(m)< {\cal R}F_n(m)$ if and only if $m < M$,
  \end{itemize}
\end{lemma}
which implies that the modified algorithm generates each word with uniform probability and has an average complexity of $O(n)$ using only numbers with $2log_2(n)$ bits\footnote{With Maple\cite{maple}, we can also compute the limit of the rate of acceptance in the main loop, we obtain $\lim_{n=\infty} C_n=2^{-3/4}=0.5946\ldots$}.

\section{Motzkin left factors with size $n$ and final height $h'$}\label{motzkin}
Many efficient algorithms are known to efficiently generate the Motzkin word with a fixed size, the simplest ones are based on a certain rejection method: Alonso\cite{ALONSO1994529}, Bacher et al.\cite{BACHER201742}, Brlek et al. \cite{Brlek2006NonUR}. Similarly, Barcucci et al.\cite{Barcucci1994TheRG} proposes an efficient algorithm based on a simple rejection method to generate Motzkin left factors of size $n$.

All these algorithms draw some sequence of $n$ small random numbers (often an integer between $1$ and $3$) and then accept/reject these sequences by doing $O(n)$ tests using numbers with $\log n$ bits. It seems difficult to extend these methods to generate a random Motzkin left factor of given size $n$ and final height $h'$ except:
\begin{itemize}
\item Brlek et al.\cite{Brlek2006NonUR} can be modified to transform a Motzkin left factor with overall height greater than or equal to $h$ into a Motzkin left factor with final height $h'$. This gives us a valid algorithm but its average complexity will remain in $O(n)$ only if $h'$ is small enough: $h'=O(\sqrt n)$,
\item Alonso\cite{ALONSO1994529} because this paper is an extension of this algorithm.
\end{itemize}

To simplify the notation, let us note $h=h'+1$. When $h$ is small enough: $h\leq\frac {3n-25} 7$, we can use the `Binomial Generator' and get a valid algorithm with average complexity in $O(n)$, but it fails when $h$ is larger. We will therefore choose the `Extended Binomial Generator' which allows us to obtain an algorithm\footnote{Of course, if we want to generate only a Motzkin word ($h=1$), this algorithm is less efficient than the others because it sometimes requires to use $2 \log n$ integer bits (instead of $\log n$ bits)} which works when $h<n-\frac 1 2-\frac{\sqrt{8n+9}}2$. Finally, to be exhaustive, we will prove that we can simply use the `Basic Generator' when $h\geq n-\frac 1 2-\frac{\sqrt{8n+9}}2$.

Indeed, using the ``Cycle Lemma''\cite{DERSHOWITZ199035}, we obtain the number of Motzkin left factors of final height $h'=h-1$ with $0<h\leq n+1$ :
\[F_n^h =\frac h {n+1} \sum_{m=0}^{\lfloor \frac{n+1-h}2\rfloor} \combi{n+1}{m,m+h} = \sum_{m=0}^{\lfloor \frac{n+1-h}2\rfloor} F_n^h(m)\]
where $F_n^h(m)=\frac h {n+1}\combi{n+1}{m,m+h}$ is the number of Motzkin left factor of final height $h-1$ that contains $n$ letters and $m$ letter $\bar x$ (ie. step $(1,-1)$). Therefore, after choosing $m$ with adequate probability, we can finish by mixing $m+h$ letters $x$ (ie. step $(1,1)$), $m$ letters $\bar x$ and $n+1-h-2m$ letters $z$ (ie. step $(1,0)$) and use the ``Lemma Cycle'' to get the final Motzkin left factor.

This give us:
\[ {\cal R}F_n^h(m) = \frac{(n+1-h-2m)(n-h-2m)}{(m+1)(m+1+h)}, \]
Now, we first need to establish some results about $F_n^h$ (this lemma is proved in the Appendix~\ref{proofMotz}):
\begin{lemma}\label{lemmaM1}
  When $n>1$, 
  \begin{itemize}
  \item there exists a unique value $\tilde m$ in $[0 \ldots \lfloor \frac{n+1-h}2\rfloor]$ such that $f(x)>0$ with $x\in [0,\lfloor \frac{n+1-h}2\rfloor]$ if and only if $x<\tilde m$,
  \item $\tilde m=\frac {2(n+1)}3 -\frac h 2 -\frac{\sqrt{4n^2+20n+28-3h^2}}6 \leq \frac {n-h} 3$, $M=\lceil \tilde m \rceil \approx \tilde m$,
  \item when $h<n-\frac 1 2 -\frac{\sqrt{8n+9}}2$, $\tilde m>1$, $M\geq 2$,
  \end{itemize}
  When $M\geq 2$, $\frac{F_n^h}{F_n^h(M)}=\Omega(\sqrt M)$.
\end{lemma}

\subsection{Generation of a Motzkin left factor when $h<n-\frac 1 2-\frac{\sqrt{8n+9}}2$}
We can now fix the values of $k$ and $\alpha$ in order to instantiate the Abstract Algorithm (and prove that these values give a valid algorithm with linear complexity).
We choose:
\begin{itemize}
\item $k$ as the smallest integer such that $(k+1)M \geq (k+1)\tilde m\geq \frac{n+1-h}2$ : $k=\lceil \frac{n+1-h}{2\tilde m} \rceil -1$,
\item $\alpha$ as the smallest integer between $0$ and $k-1$ which makes the root of ${\underline g}(x)=0$ close of $M$ : $\alpha=max(k-1,\lfloor k-k(M-\tilde m)\rfloor)$.
\end{itemize}
We have:
\begin{eqnarray*}
  {\cal R}{\underline G}^{k,\alpha}_M(m) &=& \frac{k(n+1-h-2m)(n-h-2m)}{(m+1+h)((k+1)M+\alpha-m)} \\
  {\cal R}{\overline G}^{k,\alpha}_M(m) &=& {\cal R}{\underline G}^{k,\alpha}_M(m)\text{ when }m<M-1\text{ or }m\geq M+1 \\
\end{eqnarray*}
and when $B^{k,\alpha}_M(M)=B^{k,\alpha}_M(M+1)$
\begin{eqnarray*}
  {\cal R}{\overline G}^{k,\alpha}_M(M-1)&=& \frac{k^2(n+3-h-2M)(n+2-h-2M)(M+1)}{(M+h)(k\,M+\alpha+1)(k\,M+\alpha)} \\
  {\cal R}{\overline G}^{k,\alpha}_M(M)&=& {\cal R}F_n^h(M)
\end{eqnarray*}
and when $B^{k,\alpha}_M(M)=B^{k,\alpha}_M(M-1)$
\begin{eqnarray*}
  {\cal R}{\overline G}^{k,\alpha}_M(M-1)&=& {\cal R}F_n^h(M-1)\\
  {\cal R}{\overline G}^{k,\alpha}_M(M)&=& \frac{k^2(n+1-h-2M)(n-h-2M)M}{(M+1+h)(k\,M+\alpha+1)(k\,M+\alpha)}.
\end{eqnarray*}

Using the Proposition~\ref{complexity}, we can prove that the complexity of the algorithm is on average linear in $n$. It only remains to prove that the algorithm is valid: ${\cal R}G_M^{k,\alpha}(m)>1$ iff $m<M$; this is proved with the following lemma which is proved in the Appendix~\ref{proofMotz}:
\begin{lemma}\label{lemmaM2}
  We have when $m\leq M-1$, ${\cal R}G_M^{k,\alpha}(m)>1$ and when $m\geq M$, ${\cal R}G_M^{k,\alpha}(m)\leq 1$.

  Moreever:
  \begin{itemize}
  \item when $B_M^{k,\alpha}(M)=B_M^{k,\alpha}(M+1)$, $\frac{k(n+3-h-2M)(n+2-h-2M)}{(M+h)(k\,M+\alpha)}>1$,
  \item when $B_M^{k,\alpha}(M)=B_M^{k,\alpha}(M-1)$, $\frac{k(n+1-h-2M)(n-h-2M)}{(M+1+h)(k\,M+1+\alpha)}\leq 1$
  \end{itemize}
\end{lemma}

So we can update the main loop of the previous algorithm to draw $m$ with probability $\frac {F_n^h(m)}{F_n^h}$ to get an algorithm that works when $h\leq n-\frac 1 2 -\frac{\sqrt{8n+9}}2$ in average time in $O(n)$ and we only use numbers with less than $2\log_2(n)$ bits: see Algorithms~\ref{algAcceptMotzkin} and \ref{algMotzkin}.

\begin{algorithm}[H]
\caption{Accept $m$ with probability $\frac{F_n^h(m)}{{\overline B}_M^{k,\alpha}(m)}$ when $1\leq \delta\leq k-1$:}\label{algAcceptMotzkin}
\begin{algorithmic}
  \Function{accept\_m}{$M,k,\alpha,which,m$}
  \IIf {$m>(n+1-h)/2$} reject $m$ \EndIIf
  \If {$m<M$}
  \State \Comment{accept $m$ with probability ${\cal R}{\overline B}_M^{h,\alpha}(M-1)/{\cal R}F_n^h(M-1)$}
  \If {$which=\text{'M+1'}$}
    \IIf{$random(k(M+1))\geq k\,M+\alpha+1$} reject $m$ \EndIIf
    \State $v\gets random(k(n+3-h-2M)(n+2-h-2M))$
    \IIf{$v\geq (M+h)(k\,M+\alpha)$} reject $m$ \EndIIf
  \Else
    \State $v\gets random((n+3-h-2M)(n+2-h-2M))$
    \IIf{$v\geq M(M+h)$} reject $m$ \EndIIf
  \EndIf
  \ForAll{$i \in [m,M-2]$}
  \State \Comment{accept $m$ with probability ${\cal R}{\overline B}_M^{h,\alpha}(i)/{\cal R}F_n^h(i)$}
  \State $v\gets random(k(n+1-h-2i)(n-h-2i))$
  \IIf{$v\geq (i+1+h)((k+1)M+\alpha-i)$} reject $m$ \EndIIf
  \EndFor
  \Else
  \State \Comment{accept $m$ with probability ${\cal R}{\overline B}_M^{h,\alpha}(M)/{\cal R}F_n^h(M)$}
  \If {$which=\text{'M-1'}$}
    \IIf{$random(k\,M+\alpha)\geq k\,M$} reject $m$ \EndIIf
    \State $v\gets random((M+h+1)(k\,M+\alpha+1))$
    \IIf{$v\geq k(n+1-h-2M)(n-h-2M)$} reject $m$ \EndIIf
  \Else
    \State $v\gets random((M+1)(M+1+h))$
    \IIf{$v\geq (n+1-h-2M)(n-h-2M)$} reject $m$ \EndIIf
  \EndIf
  \ForAll{$i \in [M+1,m-1]$}
  \State \Comment{accept $m$ with probability ${\cal R}F_n^h(i)/{\cal R}{\overline B}_M^{h,\alpha}(i)$}
  \State $v\gets random((i+1+h)((k+1)M+\alpha-i))$
  \IIf{$v\geq k(n+1-h-2i)(n-h-2i)$} reject $m$ \EndIIf
  \EndFor
  \EndIf
  \State accept $m$
  \EndFunction

\end{algorithmic}
\end{algorithm}

\begin{algorithm}[H]
\caption{Draw a Motzkin Left Factor with size $m$ and final height $h-1$}\label{algMotzkin}
\begin{algorithmic}[1]
  \IIf {$h\geq n+\frac 1 2- \frac{\sqrt{8n+9}}2$} FAILS \EndIIf \Comment{$h$ is too big}
  \State $M \gets \text{first value of $m$ which maximalizes $F_n^h(m)$}$
  \State $\tilde m \gets \frac{4n+4-3h-\sqrt{4n^2+20n+28-3h^2}}6$
  \State $k \gets \lceil\frac{n+1-h}{2\tilde m}\rceil-1$
  \State $\alpha \gets max(k-1,k-\lceil k(M-\tilde m)\rceil)$
  \If {$(k\,M+\alpha+1)(k\,M+\alpha)\geq k^2M(M+1)$}
  \State $which=\text{'M+1'}$
  \Else
  \State $which=\text{'M-1'}$
  \EndIf
  \While{true}
  \State \Comment{draw $m$ with probability $\frac {B_M^{k,\alpha}(m)} {\sum_{i=0}^{(k+1)M}B_M^{k,\alpha}(i)}$}
  \State $m \gets \textsc{extended\_bin}(M,k,\alpha)$
  \State \Comment{accept $m$ with probability $F_n^h(m)/{\overline B}_M^{k,\alpha}(m)$}
  \IIf {$\textsc{accept\_m}(M,k,\alpha,which,m)$} break \EndIIf
  \EndWhile
  \State \Comment{draw a Motzkin left factor of size $n$ with $m+h-1$ letters $x$ and $m$ letters $\bar x$}
\end{algorithmic}
\end{algorithm}
We can verify that, except when we check if $h$ is small enough (line $1$) or when we compute $k$ and $\alpha$ (lines 3-5), we use only $2\log_2(n)$ bits numbers. However, it seems more difficult to completely avoid using some floating-point computations in the pre-computation phase of the algorithm (lines 1-5)\footnote{The most problematic problem is the computation of $k$. In fact, it seems possible (as $M\geq 2$) to replace this computation by $k\gets \lceil\frac{n+1-h}{2(M-1)}\rceil-1$ but this would mean completely rewriting the proofs of Lemma~\ref{lemmaM2}. Concerning $\alpha$, it is enough to find a value of $\alpha\ in [0,k-1]$ such that $\max({\cal R}G(M), {\cal R}G(M+1)) \leq 1 \leq \min({\cal R}G(M-2), {\cal R}G(M-1))$, which can be done by testing all potential values in $O(k)$.}.

\subsection{Generation of a Motzkin left factor when $h\geq n-\frac 1 2-\frac{\sqrt{8n+9}}2$}

In this case, we know that $M$ is equal to $0$ or $1$ ; we need to find a value $m$ in $[0,\lfloor \frac{n+1-h} 2\rfloor]$ with an adequate probability: $\frac {F_n^h(m)}{F_n^h}$.

Let us use the `Basic Generator' with $k=\lfloor \frac{n+1-h} 2\rfloor\leq \frac 1 4+\frac{\sqrt{8n+9}}4=O(\sqrt n)$. We have ${\cal R}G_M^k(m)={\cal R}F_n^h(m)$ and therefore ${\cal R}G_M^k(m)>1$ if and only if $m<M$, so the algorithm is valid.

We can apply Proposition~\ref{complexity} with $C_1=O(1)$, $C_2=O(k)$, $\frac{F(M)}{\sum_m F(m)}=O(1)$, $\frac{\sum_m B(m)}{B(M)}=k+1$ to prove that the total complexity of the main loop is, on average, $O(1+k^2)=O(n)$ as desired.

\section{Discussion}
Many recent generation algorithms use rejection methods to efficiently generate simple structures. Most of them are based on clever one-to-one correspondance. Here we try to look for methods that try to break the problem into small parts: find some $m$ with adequate probability and get a simple problem to solve. As far as we know, this method has been used to generate Motzkin words, partial injection, colored unary-binary trees~\cite{GouyouBeauchamps2010RandomGU}.

After adapting differently this method to Fibonacci numbers, we obtain an efficient and very easy to implement algorithm and were surprised to find that we can use it with minor modifications to generate Schroëder trees\footnote{In fact, it seems to give a $O(n)$ algorithm for many similar problems: Motzkin words, binary tree forest, ternary tree forest, ... which are not explained in this article and also a $O(n^{5/4})$ algorithm for partial injection}. To study the limits of this kind of algorithms, we try to generate Motzkin left factors with final height $h$: a problem that seems more complicated because the number of words varies a lot: from $M_n^0\approx \frac{3^n}{n\sqrt n}$ to $M_n^n=1$ \footnote{also note that we have $\ln(M_n^{n-\sqrt{2n}}) \approx \frac {\sqrt 2} 2 \sqrt n\ln n$ as the `Extended Binomial Generator' only works when $h\leq n-\frac 1 2-\frac{\sqrt{8n+9}}2$} ; we needed to extend the 'Binomial Generator' but we obtained a final algorithm which is efficient and can still be implemented relatively easily.

Note also that, as in~\cite{GouyouBeauchamps2010RandomGU}, we looked for generators that attempt to reproduce the desired distribution: same value of $M$ that maximizes the initial and final distributions, and similar "dispersion" (within a constant factor) and we found that at least for these three problems, there are some that are really simple. But is this a good idea or a misleading one? That will be the subject of a future article\cite{AlonsoRandomII}. Let's just note here that this article will allow to implement the `Binomail generator' and the `Extended binomial generator' more efficiently (and thus to decrease the average complexity of the algorithms presented in the article).

What more can we say? Perhaps, that the problems presented in this article correspond to numbers whose asymptotics are in $l(n)C^n$ where $C$ is a solution of a polynomial of degree $2$; this probably explains why we can solve them using numbers with only $2\,log_2(n)$ bits. When $C$ is a solution of a polynomial of degree $3$, $\ldots$, this method (when valid) will probably require using numbers of $3\,log_2(n)$ bits, $\ldots$.

Can we do better ? For example, can we find an algorithm that uses only integers smaller than $n$ to generate a Fibonacci word of 'size' $n$ ?

\paragraph{Comparaison with Gouyou-Beauchamps et all~\cite{GouyouBeauchamps2010RandomGU}}
In~\cite{GouyouBeauchamps2010RandomGU}, they present a generic algorithm for drawing a number $m$ per rejection, so that the distributions we obtained `automatically' can be used to define their initial distribution ${\overline B}_n(m)$, a distribution that allows us to find efficient algorithms for accepting $m$ with probability $\frac{F_n(m)}{{\overline B}_n(m)}$. In fact, we can think of our algorithm as an instantiation of \cite{GouyouBeauchamps2010RandomGU} that uses a fixed method to accept a value of $m$ with probability $\frac{F_n(m)}{{\overline B}_n(m)}$. Fixing this method allows us to quickly guess whether we have a chance of finding an efficient algorithm for certain structures by choosing a generator and checking (for certain values of $n$) whether ${\cal R}{\overline B}_n(m)<{\cal R}F_n(m)$ iff $m<M$. Of course, this also restricts the potential choices of ${\overline B}_n(m)$ but, at least for the structures we examined, we were able to quickly obtain an efficient algorithm.

We may also note that we pay a price: even though we get similar complexity, the proofs are more complicated and the algorithm used to accept a value of $m$ when generating a Motzkin left factor with $h=0$ is less elegant than the algorithm presented in \cite{ALONSO1994529} ; this will also be true if we try to instantiate our method with the `Binomial Generator' to generate partial injections and compare it with \cite{GouyouBeauchamps2010RandomGU}\footnote{For this problem, it is better to use the `Extended Binomial Generator' with $k=\lceil\sqrt n\rceil+1$ to choose $n-m$, which gives a valid $O(n)$ complexity algorithm that is valid when $n\geq 2$.}.

Finally, notice that there is a slight difference in the distribution used to guess the initial value of $m$,
\begin{itemize}
\item we choose standard binomial distributions, but we modify them slightly by decreasing the value of ${\overline B}_n(M)$ (to force ${\overline B}_n(M)={\overline B}_n(M\pm1)$),
\item in \cite{GouyouBeauchamps2010RandomGU}, they use the standard binomial distribution which explains that they need to choose a value of ${\overline B}_n(M)=$ which is sometimes larger than $F_n(M)$.
\end{itemize}

\bibliographystyle{alpha}
\bibliography{random_generation}

\clearpage
\appendix
\section{Appendices: proofs}
\subsection{Appendix: `Extended Binomial Generator's proof}\label{proofG2}
\begin{proof}{Proof of lemma~\ref{lemmaG2}}
  We have:
  \[ \frac{{\underline B}_M^{k,\alpha}(m+1)}{{\underline B}_M^{k,\alpha}(m)}=\frac{(k+1)M+\alpha-m}{k(m+1)} \leq 1\]
  iff $(k+1)M+\alpha-m\leq k(m+1)$ iff $M+\frac{\alpha-k}{k+1}\leq m$ iff $M\leq m$. 
  
  Let us first prove that the relations are valid for ${\underline B}_M$. Let us note:
  \[l(\alpha)=\frac{{\underline B}_M^{k,\alpha}(M)}{\sum_m {\underline B}_M^{k,\alpha}(m)}=\frac{((k+1)M+\alpha)!k^{k\,M+\alpha}}{M!(k\,M+\alpha)!(k+1)^{(k+1)M+\alpha}}.\]

  Using the Stirling Formula, we get:
  \[ l(0)=\frac{((k+1)M)!k^{k\,M}}{M!(k\,M)!(k+1)^{(k+1)M}} \approx \frac {\sqrt{2(k+1)}}{2\sqrt{\pi\,k\,M}}\]
  and more precisely by using $\sqrt{2\pi\,n}\left(\frac n e\right)^n\leq n!\leq \sqrt{2\pi\,n}\left(\frac n e\right)^n \exp\left(\frac 1 {12n}\right)$:
  \begin{eqnarray*}
    0.84 &\leq& exp\left(-\frac 1 6\right)\leq exp\left(-\frac 1 {12\,M}-\frac 1 {12\,k\,M}\right) \leq \frac{2\sqrt{\pi\,k\,M}l(0)}{\sqrt{2(k+1)}} \\
    & \leq & exp\left(\frac 1 {12(k+1)M}\right) \leq exp\left(\frac 1 {24}\right) \leq 1.05.
  \end{eqnarray*}
  We also have:
  \[\frac{l(\alpha)}{l(\alpha-1)}=\frac{k((k+1)M+\alpha)}{(k+1)(k\,M+\alpha)}=1-\frac{\alpha}{(k+1)(k\,M+\alpha)}\]
  therefore we obtain: $l(\alpha)\leq l(0)$ and
  \[\frac{l(\alpha)}{l(0)}\geq 1-\sum_{i=1}^{\alpha} \frac i {(k+1)(k\,M+i)}\geq 1 - \frac{\alpha(\alpha+1)}{2 (k+1)(k\,M+1)}\geq 1 - \frac{k}{2(k\,M+1)}>\frac 1 2.\]

  Finally, note that $n=(k+1)M+\alpha$, with a constant factor $(k+1)^n$, ${\underline B}_M$ corresponds to the binomial distribution with $p=\frac 1 {k+1}$, $q=1-p$. The variance of ${\underline B}_M$ is therefore $n,p,q=O\left(\frac{(M+1)(k+1)k}{(k+1)^2}\right)=O(M+1)$. Using Jensen's inequality, we obtain: $\frac{\sum_m |m-M|{\underline B}_M^{k,\alpha}(m)}{\sum_m {\underline B}_M^{k,\alpha}(m)}=O(\sqrt M)$.

  We end by noting that these relations remain valid for $B_M$ as $B_M(M)=\frac{kM}{kM+\alpha+1}{\underline B}_M(M)$ or $B_M(M)=\frac{kM+\alpha}{k(M+1)}{\underline B}_M(M)$ and thus $B_M(M)={\underline B}_M(M)(1-O(1/M))$.
\end{proof}
\subsection{Appendix: Fibonacci's proofs}\label{proofFib}
\begin{proof}{Proof of Lemma~\ref{lemma1}}
  Basic computations give $f'(x)=-5n+10x+3$, $f(0)=n^2-2n\geq 0$ if $n>1$, $f(\lfloor n/2\rfloor)=-\lceil n/2\rceil(\lfloor n/2\rfloor+1)<0$ if $n\geq 1$. This implies that there exists a unique value of $x\in [0,\lfloor n/2\rfloor]$: $\tilde m$ such that $f(x)>0$ if and only $x<\tilde m$.

  Using ${\cal R}F_n(x)>1$ if and only if $f(x)>0$ with ${\cal R}F_n(m)=F_n(m+1)/F_n(m)$, we get $M=\lceil \tilde m \rceil$. Solving $f(\tilde m)=0$ in $[0,\lfloor n/2\rfloor]$ gives:
  \[\tilde m=\frac {5n-3-\sqrt{5n^2+10n+9}}{10}=\frac{5-\sqrt 5}{10}n-\frac{3+\sqrt{5}}{10}+O(1/n) \approx \frac{5-\sqrt 5}{10}n.\]

  Suppose that $\epsilon\geq 0$ and $n>2$, we have:
  \begin{eqnarray*}
    {\cal R}F_n(\tilde m+\epsilon)&=&\frac{(n-2\tilde m)(n-1-2\tilde m)}{(\tilde m+1)(n-\tilde m)}\frac{(1-2\frac{\epsilon}{n-2\tilde m})(1-2\frac{\epsilon}{n-1-2\tilde m})}{(1+\frac{\epsilon}{\tilde m+1})(1-\frac{\epsilon}{n-\tilde m})}\\
    & \geq & 1 (1-2\frac{\epsilon}{n-2\tilde m}-2\frac{\epsilon}{n-1-2\tilde m}-\frac{\epsilon}{\tilde m+1}+0) \geq 1-\frac{5\epsilon}{\tilde m}
  \end{eqnarray*}
  by noting that $\tilde m\leq \frac{n-1} 3$ (and so $n-1-2\tilde m\geq \tilde m$).
  Therefore, when $i\geq 0$, we get ${\cal R}F_n(M+i)\geq 1-\frac{5\,i}{\tilde m}$ and when $i<\frac {\tilde m}5$, $\frac{F_n(M+i)}{F_n(M)}\geq \prod_{j=0}^{i-1}(1-\frac{5\,(j+1)}{\tilde m}) \geq 1-\frac{5i(i+1)}{2 \tilde m}$. This implies that when $0\leq i \leq \frac{\sqrt{\tilde m-1}}{\sqrt 5}$, $F_n(M+i)\geq \frac 1 2 F(M)$ and thus $F_n=\Omega(\sqrt M)F_n(M)=\Omega(\sqrt n)F_n(M)$.
\end{proof}
\begin{proof}{Proof of Lemma~\ref{lemma2}}
  We have ${\cal R}{\overline B}_M(M-1)={\cal R}{\overline B}_M(M)=1$ and using the previous lemma, ${\cal R}F_n(M-1)>1$, ${\cal R}F_n(M)\leq 1$. This validates the inequality for $m=M-1$ and $m=M$.
  
  For other values of $m$, we have ${\cal R}{\overline B}_M(m)< {\cal R}F_n(m)$ if and
  only if ${\underline g}(m)=(n-2m)(n-2m-1)-(2M-m)(n-m)>0$.  We have ${\underline g}(0)=n(n-2M-1)>0$,
  ${\underline g}(n/2-1)=-(2M+1-n/2)(n/2-1)<0$ when $n>2$,
  ${\underline g}'(x)=6x-3n+2M+2$. This implies that when $x\in [0,n/2-1]$, there exists a value $\tilde x$
  such that ${\underline g}(x)\geq 0$ iff $x\leq
  \tilde x$.

  It is therefore sufficient to prove that ${\underline g}(M-2)>0$ and ${\underline g}(M+1)\leq 0$.

  Using $M=\lceil \tilde m\rceil$, we obtain (as $\frac {\delta {\underline g}(M-2)}{\delta M}<0$):
  \begin{eqnarray*}
    {\underline g}(M-2)&=&5M^2-(5n+14)M+n^2+5n+8 \\
    & \geq & 5(\tilde m+1)^2-(5n+14)(\tilde m+1)+n^2+5n+8 \\
    &=&\frac{11}{10}-\frac {3n}2+\frac{7\sqrt{5n^2+10n+9}}{10}>0.
  \end{eqnarray*}

  Similarly, we get (as $\frac {\delta {\underline g}(M+1)}{\delta M}<0$ when $n>6$): 
  \begin{eqnarray*}
    {\underline g}(M+1)&=&5M^2-(5n-10)M+n^2-4n+5 \\
    & \leq & 5\tilde m^2-(5n-10)\tilde m+n^2-4n+5 \\
    & = & \frac{29}{10}+\frac {3n}2-\frac{7\sqrt{5n^2+10n+9}}{10},
  \end{eqnarray*}
  which is less than or equal to 0 when $n\geq 20$.

  Finally, we can calculate the value of ${\underline g}(M+1)$ when $n<20$, we find:
  \begin{tabular}{|l|c|c|c|c|c|c|c|c|c|c|c|c|c|c|c|c|c|c|c|c|}
    \hline
    $n$ & 0 & 1 & 2 & 3 & 4 & 5 & 6 & 7 & 8 & 9 \\
    \hline
    $M$ & 0& 0& 0& 1& 1& 1& 2& 2& 2& 2 \\
    \hline
    ${\underline g}(M+1)$ & 5& 2& 1& 2& 0& 0& -3& -4& -3& 0\\
    \hline
  \end{tabular} \\
  \begin{tabular}{|l|c|c|c|c|c|c|c|c|c|c|}
    \hline
    $n$ & 10 & 11 & 12 & 13 & 14 & 15 & 16 & 17 & 18 & 19 \\
    \hline
    $M$ & 3& 3& 3& 4& 4& 4& 4& 5& 5& 5 \\
    \hline
    ${\underline g}(M+1)$ & -10& -8& -4& -18& -15& -10& -3& -24& -18& -10 \\
    \hline
  \end{tabular} \\
  which ends the proof by noting that when $n\leq 3$, there exist no word with $M+2$ letters $|b|$.
\end{proof}
\subsection{Appendix: Schröder's proof}\label{proofS1}
\begin{proof}{Proof of Lemma~\ref{lemmaS1}}
  Indeed, ${\cal R}F_n(x)>1$ if and only if $f(x)=(n-x)(n+1+x)-(x+1)(x+2)>0$. Basic computations give $f'(x)=-4x-4$, $f(0)=n(n+1)-2\geq 0$ if $n\geq 1$, $f(n)=-(n+1)(n+2)<0$. This implies that exists a unique value of $x\in [0,n]$: $\tilde m$ such that $f(x)>0$ if and only $x<\tilde m$.

  The relation ${\cal R}F_n(x)>1$ if and only if $f(x)>0$ with ${\cal R}F_n(m)=F_n(m+1)/F_n(m)$ gives $M=\lceil \tilde m \rceil$. By solving $f(\tilde m)=0$ in $[0,n]$, we obtain:
  \[\tilde m=-1+\frac {\sqrt{2n^2+2n}}{2}=\frac n {\sqrt 2}-\frac {4-\sqrt 2} 4+O(1/n) \approx \frac n {\sqrt 2}.\]
  We can also note that $\tilde m<\frac n {\sqrt 2}$ (and so $n-\tilde m\geq (\sqrt 2 -1)\tilde m)$.
  
  To prove that $F_n=\Omega(\sqrt n)F_n(M)$, we can imitate the proof of Lemma~\ref{lemma1}. Indeed, assuming that $\epsilon\geq 0$, we have :
  \[{\cal R}F_n(\tilde m+\epsilon)=\frac{(1-\frac{\epsilon}{n-\tilde m})(1-\frac{\epsilon}{n+1-\tilde m})}{(1+\frac{\epsilon}{\tilde m+1})(1+\frac{\epsilon}{\tilde m+2})}\geq 1 -\frac{7\epsilon}{\tilde m}.\]
  We obtain therefore when $0\leq i<\frac{\tilde m}7$, $\frac{F_n(M+i)}{F_n(M)}\geq 1-\frac{7\,i(i+1)}{2\tilde m}$, which gives us: $F_n=\Omega(\sqrt n)F_n(M)$.

  We have ${\cal R}{\overline B}_M(M-1)={\cal R}{\overline B}_M(M)=1$ and ${\cal R}F_n(M-1)>1$, ${\cal R}F_n(M)\leq 1$. This validates the last inequality for $m=M-1$ and $m=M$.
  
  For other values of $m$, we have ${\cal R}{\overline B}_M(m)< {\cal R}F_n(m)$ if and
  only if ${\underline g}(m)=(n-m)(n+1+m)-(2M-m)(m+2)>0$. We have ${\underline g}(0)=n^2+n-4M>0$, ${\underline g}(n)=-(2M-n)(n+2)<0$ (when $n>3$), ${\underline g}'(x)=1-2M$. This implies that when $x\in [0,n]$, there exists a value $\tilde x$
  such that ${\underline g}(x)\geq 0$ iff $x\leq \tilde x$. It is therefore sufficient to prove that ${\underline g}(M-2)>0$ and ${\underline g}(M+1)\leq 0$.

  Using $M=\lceil \tilde m\rceil$, we obtain
  \begin{eqnarray*}
    {\underline g}(M-2)&=&-2M^2+n^2+M+n-2>-2(\tilde m+1)^2+n^2+(\tilde m+1)+n-2\\
    &=&\frac{\sqrt{2n^2+2n}-4}2
  \end{eqnarray*}
  which is greater than $0$ when $n\geq 3$.

  Similarly,
  \[{\underline g}(M+1)=-2M^2+n^2-5M+n+1\leq -2\tilde m^2+n^2-5\tilde m+n+1=\frac{8-\sqrt{2n^2+2n}}2\]
  which is less than or equal to $0$ when $n\geq 6$.

  Finally, we can compute the values of ${\underline g}(M-2)$ and ${\underline g}(M+1)$ when $n<6$, we find:
  
  \begin{tabular}{|l|c|c|c|c|c|c|}
    \hline
    $n$ & 0 & 1 & 2 & 3 & 4 & 5 \\
    \hline
    $M$ & -1 & 0& 1& 2& 3& 3  \\
    \hline
    ${\underline g}(M-2)$ & -5& 0 & 3 & 4& 3& 13\\
    \hline
    ${\underline g}(M+1)$ & 4& 3& 0& -5& -12& -2 \\
    \hline
  \end{tabular} \\
\end{proof}
\subsection{Appendix: Motzkin left factor's proofs}\label{proofMotz}
\begin{proof}{Proof of Lemma~\ref{lemmaM1}}
  Indeed, ${\cal R}F_n^h(x)>1$ if and only if $f(x)=(n-h-2x)(n+1-h-2x)-(x+1)(x+1+h)>0$. Basic computations give $f'(x)=-4n+6x+3h-4$, $f(0)=(n-h)^2+n-2h-1>0$ if $n\geq 1$, $f(\lfloor \frac {n+1-h}2\rfloor)=-(\lfloor \frac {n+1-h}2\rfloor+1)(\lfloor \frac {n+1-h}2\rfloor+1+h)<0$. This implies that exists a unique value of $x\in [0,\lfloor\frac{n+1-h}2\rfloor]$: $\tilde m$ such that $f(x)>0$ if and only $x<\tilde m$.

  The relation ${\cal R}F_n^h(x)>1$ if and only if $f(x)>0$ with ${\cal R}F_n^h(m)=F_n^h(m+1)/F_n^h(m)$ gives $M=\lceil \tilde m \rceil$. Solving $f(\tilde m)=0$ in $[0,n]$ gives:
  \[\tilde m=\frac {2(n+1)}3 -\frac h 2 -\frac{\sqrt{4n^2+20n+28-3h^2}}6 \leq \frac{n-h}3 \leq \frac{n+1-h}3,\]
  so $\tilde m=1$ corresponds to $h=n-\frac 1 2 + \frac{\sqrt{8n+9}}2$ which implies $M\geq 2$ when $n-\frac 1 2 + \frac{\sqrt{8n+9}}2$ because $\frac{\delta \tilde m}{\delta h}<0$.



  Now suppose that $\epsilon\geq 0$, we have:
  \[{\cal R}F_n^h(\tilde m+\epsilon)=\frac{\left(1-\frac{2\epsilon}{n-h-2\tilde m}\right)\left(1-\frac{2\epsilon}{n+1-h-2\tilde m}\right)}{\left(1+\frac{\epsilon}{\tilde m+1}\right)\left(1+\frac{\epsilon}{\tilde m+h+1}\right)} \geq 1 -\frac{6\epsilon}{\tilde m} \]
  so when $0\leq i < \frac{\tilde m}6$, we get $\frac{F_n^h(M+i)}{F_n^h(M)}\geq 1-\frac{3i(i+1)}{\tilde m}$ which gives us $F_n^h=F_n^h(M)\Omega(\sqrt{\tilde m})=F_n^h(M)\Omega(\sqrt M)$ when $M\geq 2$.
\end{proof}
\begin{proof}{Proof of Lemma~\ref{lemmaM2} when $m\leq M-2$ or $m\geq M+1$}
  We have ${\underline g}(-h-1)=k(n+3+h)(n+2+h)>0$ and 
  \[{\underline g}\left(\frac{n-h}2\right)=-\frac{n+h+2}2\left((k+1)M+\alpha- \frac{n-h}2\right)<0\]
  as $(k+1)M\geq (k+1)\tilde m\geq \frac{n+1-h}2$.
  
  We have:
  \[{\underline g}'(x)=(8k+2)x-(4n+2-4h+M)k-M-\alpha+h+1,\]
  so ${\underline g}$ has a unique root $\tilde x$ in $\left[-h-1,\frac{n-h}2\right[$, and we have ${\underline g}(x)>0$ if $x<\tilde x$ and ${\underline g}(x)<0$ if $x>\tilde x$.

  We have $\alpha=max(k-1,\lfloor k-k(M-\tilde m)\rfloor)$, so:
  \[ (k+1)\tilde m+k-1\leq M+k(\tilde m+1)-1\leq (k+1)M+\alpha\leq M+k(\tilde m+1)< (k+1)\tilde m+k+1.\]
  But:
  \begin{eqnarray*}
    {\underline g}(\tilde m)&=&k\,f(\tilde m)-(\tilde m+1+h)((k+1)M+\alpha-\tilde m)+k(\tilde m+1+h)(\tilde m+1)\\
    &=&(\tilde m+1+h)((k+1)\tilde m+k-((k+1)M+\alpha))
  \end{eqnarray*}
  and
  \[ g_1(x)={\underline g}(x+1)-{\underline g}(x)=-((k+1)M+\alpha)-2(2n-2h-4x-1)+h+2x+2.\]
  Therefore when $h<n-1$ (which is always true) as $4h-4n+7\tilde m+1<0$:
  \begin{eqnarray*}
    g(\tilde m+1)&=&g(\tilde m)+g_1(\tilde m)\leq (4h-4n+7\tilde m+1)k+2h+2\tilde m+4\\
    &\leq & \frac{-10\tilde m^2+(15n-11h+13)\tilde m-(4n-4h-1)(n-h+1)+4\,f(\tilde m)}{2\tilde m}\\
    &=& \frac{2\tilde m^2 + (h - n - 3)\tilde m - 5h + n - 3}{2\tilde m}
  \end{eqnarray*}
  which is negative when $n\geq 9$, because $\tilde m\leq \frac{n-h}3 \leq \frac{n-h+3+\sqrt{(n-h)^2+34h-2n+33}}4$. Finally, when $n\leq 8$, we check that $2\tilde m^2 + (h - n - 3)\tilde m - 5h + n - 3<0$ for all possible values of $h$.

  Similarly, as $4n-4h-7\tilde m+7>0$, we have:
  \begin{eqnarray*}
    g(\tilde m-1)&=&g(\tilde m)-g_1(\tilde m-1)\geq (4n-4h-7\tilde m+7)k-2h-2\tilde m\\
    &\geq & \frac{-4\tilde m^2-(7n-3h+7)\tilde m+(4n-4h+7)(n+1-h)-4\,f(\tilde m)}{2\tilde m}\\
    &=& \frac{-16\tilde m^2+9(n-h+1)\tilde m-3h+7n+11}{2\tilde m}\\
  \end{eqnarray*}
  which is positive because $\tilde m\leq \frac{n-h}3 \leq \frac {9n-9h+9+ \sqrt{81 (n-h)^2 - 354 h + 610n + 785}}{32}$.
  
  So we get ${\underline g}(M-2)\geq {\underline g}(\tilde m-1)>0$ and ${\underline g}(M+1)\leq {\underline g}(\tilde m+1)<0$. This is equivalent to ${\cal R}G_M^{k,\alpha}(m)>1$ if $m\leq M-2$ and ${\cal R}G_M^{k,\alpha}(m)<1$ if $m\geq M+1$.
\end{proof}
\begin{proof}{Proof of Lemma~\ref{lemmaM2} when $m=M$ or $m=M+1$}
  Note first that when $B_M^{k,\alpha}(M)=B_M^{k,\alpha}(M-1)$, ${\cal R}G_M^{k,\alpha}(M-1)={\cal R}F_n^h(M-1)>1$ while when $B_M^{k,\alpha}(M)=B_M^{k,\alpha}(M+1)$, ${\cal R}G_M^{k,\alpha}(M)={\cal R}F_n^h(M)\leq 1$.

  We also have as $0\leq \alpha \leq k-1$:
  \[ \frac{k^2(n+3-h-2M)(n+2-h-2M)(M+1)}{(M+h)(k\,M+\alpha+1)(k\,M+\alpha)} \geq \frac{k(n+3-h-2M)(n+2-h-2M)}{(M+h)(k\,M+\alpha)}, \]
  and:
  \[ \frac{k^2(n+1-h-2M)(n-h-2M)M}{(M+h+1)(k\,M+\alpha+1)(k\,M+\alpha)} \leq \frac{k(n+1-h-2M)(n-h-2M)}{(M+h+1)(k\,M+1+\alpha)}, \]
  which turn into equalities when $k=1$ and $\alpha=0$.
  Moreover, when $k=1$, we have $\alpha=0$ and $B_M^{k,\alpha}(M)=B_M^{k,\alpha}(M-1)=B_M^{k,\alpha}(M+1)$ ; we can thus conclude directly.

  We can now note that $\alpha=max(k-1,\lfloor k-k(M-\tilde m)\rfloor)$ implies that $\tilde m + \frac{k-1-\alpha}k \leq M\leq \tilde m + \frac{k-\alpha}k$.

  When $B_M^{k,\alpha}(M)=B_M^{k,\alpha}(M+1)$ (which is equivalent to $(kM+\alpha+1)(kM+\alpha)\geq k^2M(M+1)$ and implies $\alpha>0$), we have:
  \[ {\cal R}{\overline G}^{k,\alpha}_M(M-1) \geq \frac{k(n+3-h-2M)(n+2-h-2M)}{(M+h)(k\,M+\alpha)}\]
  which is greater than $1$ if $g_1(M)=k(n+3-h-2M)(n+2-h-2M)-(M+h)(k\,M+\alpha)>0$.  We have $g_1'(x)=(3h-4n+6x-10)k-\alpha<0$ if $x<\frac{n+1-h}2+\frac{k(n+7)+\alpha}{6k}$.
  Therefore:
  \begin{eqnarray*}
    g_1(M)&\geq&g_1\left(m+\frac{k-\alpha}k\right)\\
    &=&\alpha\frac{k(7\sqrt{-3h^2+4n^2+20n+28}-3h-4n-10)+24\alpha} {6k}\\
    &\geq & \alpha\frac{7\sqrt{-3h^2+4n^2+20n+28}-3h-4n-10} 6 >0
  \end{eqnarray*}
  when $h\leq n+1 < \frac{-2n-5+7\sqrt{16n^2+80n+113}}{26}$.

  When $B_M^{k,\alpha}(M)=B_M^{k,\alpha}(M-1)$ (which implies $\alpha<k-1$), we have:
  \[ {\cal R}{\overline G}^{k,\alpha}_M(M) \leq \frac{k(n+1-h-2M)(n-h-2M)}{(M+h+1)(k\,M+1+\alpha)}\]
  which is less than or equal to $1$ if $g_1(M)=k(n+1-h-2M)(n-h-2M)-(M+1+h)(k\,M+\alpha+1)\leq 0$.  We have $g_1'(x)=(3h-4n+6x-3)k-\alpha-1<0$ if $x<\frac{n+1-h}2+\frac{k\,n+\alpha+1}{6k}$.
  Therefore:
  \begin{eqnarray*}
    g_1(M)&\leq&g_1\left(m+\frac{k-\alpha-1}k\right)\\
    &=&-(k-\alpha-1)\frac{k(7\sqrt{-3h^2+4n^2+20n+28}-3h-4n-34)+24+24\alpha} {6k}\\
    &\leq & -(k-\alpha-1)\frac{7\sqrt{-3h^2+4n^2+20n+28}-3h-4n-34} 6<0
  \end{eqnarray*}
  when $h\leq n+1<\frac{-2n-17+7\sqrt{16n^2+64n+25}}{26}$ and $n>0$.
\end{proof}

\end{document}